\documentclass[letter,journal,twocolumn,twoside,times,10pt]{IEEEtran}
\usepackage{epsfig,amsfonts,subfigure,color}
\usepackage[nolist]{acronym}
\usepackage{graphicx,cite,amssymb,amsmath,mathrsfs,perso,stfloats}
\usepackage[usenames,dvipsnames]{xcolor}
\usepackage{bbm}

\newtheorem{prop}{Proposition}
\newtheorem{lemma}{Lemma}

\newcommand{\mc}{\mathcal}
\newcommand{\ve}[1]{\underline{#1}}

\begin{document}

\begin{acronym}
\acro{AWGN}{additive white Gaussian noise} \acro{EXIT}{extrinsic
information transfer} \acro{RA}{random access}
\acro{CCw/oFB}{collision channel without
feedback}\acro{SIC}{successive interference
cancelation}\acro{SN}{sum node}\acro{PN}{packet node}\acro{CSA}{coded
slotted ALOHA}\acro{MAC}{medium access control}\acro{CRA}{coded
random access}\acro{SPC}{single parity-check}\acro{MDS}{maximum
distance separable}\acro{RS}{Reed-Solomon}\acro{SNR}{signal-to-noise
ratio}\acro{D-GLDPC}{doubly-generalized LDPC}\acro{DE}{differential
evolution}\acro{PLR}{packet loss rate}\acro{LDPC}{low-density
parity-check}\acro{DVB}{Digital Video Broadcasting}\acro{RCS}{Return
Channel via Satellite}\acro{BEC}{binary erasure
channel}\acro{MAP}{maximum-a-posteriori}\acro{GJE}{Gauss-Jordan
elimination}\acro{BP}{belief
propagation}\acro{IT}{iterative}\acro{CRDSA}{contention resolution
diversity slotted Aloha}\acro{IRSA}{irregular repetition slotted
Aloha}\acro{GA-MAP}{genie-aided maximum-a-posteriori}\acro{i.i.d.}{independent and identically distributed}\acro{r.v.}{random variable}\acro{SA}{slotted Aloha}\acro{DSA}{diversity slotted Aloha}\acro{DN}{diversity node}\acro{p.m.f.}{probability mass function}
\end{acronym}

\title{The Throughput of Slotted Aloha with Diversity}
\author{Andrea Munari, Michael Heindlmaier, Gianluigi Liva and Matteo Berioli
\thanks{
Andrea Munari, Gianluigi Liva and Matteo Berioli are with Institute of Communication and Navigation of the Deutsches Zentrum f\"{u}r Luft-
und Raumfahrt (DLR), D-82234 Wessling, Germany (e-mail: \{andrea.munari,gianluigi.liva,matteo.berioli\}@dlr.de).}
\thanks{
Michael Heindlmaier is with the Lehrstuhl f\"{u}r Nachrichtentechnik, Technische Universit\"{a}t M\"{u}nchen (TUM), D-80290 M¨unchen Germany (e-mail: michael.heindlmaier@tum.de).
}
}%
\maketitle
\thispagestyle{empty} \pagestyle{empty}

%
\begin{abstract}
In this paper, a simple variation of classical Slotted Aloha is introduced and analyzed. The enhancement relies on adding multiple receivers that gather different observations of the packets transmitted by a user population in one slot. For each observation, the packets transmitted in one slot are assumed to be subject to independent on-off fading, so that each of them is either completely faded, and then does not bring any power or interference at the receiver, or it arrives unfaded, and then may or may not, collide with other unfaded transmissions.  With this model, a novel type of diversity is introduced to the conventional SA scheme, leading to relevant throughput gains already for moderate number of receivers. The analytical framework that we introduce allows to derive closed-form expression of both throughput and packet loss rate an arbitrary number of receivers, providing interesting hints on the key trade-offs that characterize the system. We then focus on the problem of having receivers forward the full set of collected packets to a final gateway using the minimum possible amount of resources, i.e., avoiding delivery of duplicate packets, without allowing any exchange of information among them. We derive what is the minimum amount of resources needed and propose a scheme based on random linear network coding that achieves asymptotically this bound without the need for the receivers to coordinate among them.
\end{abstract}

\section{Introduction}\label{sec:intro}

A renewed interest for Aloha-like \ac{RA} protocols led recently to the development of new high-throughput uncoordinated multiple-access schemes \cite{DeGaudenzi07:CRDSA,Giannakis07:SICTA,Gollakota2008:ZigZag,Liva11:IRSA,Paolini11:CSA_ICC,Kissling11:CRA,Dimakis11:SigSag,Liva2012:ISIT_ConvIRSA,Popovski2012:Aloha,Pfister2012:IRSA}. These schemes share the feature of cancelling the interference caused by
 a packet whenever (a portion of) it is successfully decoded. Among the aforementioned works, a specific class is  based on the \ac{DSA} protocol introduced in \cite{Rappaport83:DSA} enhanced by \ac{SIC}. In
 \cite{Liva11:IRSA,Paolini11:CSA_ICC} it was shown that the \ac{SIC}
 process can be well modeled by means of a bipartite graph.
 By exploiting the graph model, a
 remarkably-high capacity (e.g., up to $0.8\,\mathrm{[packets/slot]}$)
 can be achieved in practical implementations, whereas for large \ac{MAC} frames it was demonstrated that fully efficiency ($1\,\mathrm{[packets/slot]}$) can be substantially attained \cite{Paolini11:CSA_Globecom,Liva2012:ISIT_ConvIRSA,Pfister2012:IRSA}.
A further key ingredient to attain large throughput gains deals with the exploitation of \emph{diversity}. As an example, the approaches proposed in  \cite{Rappaport83:DSA,DeGaudenzi07:CRDSA,Liva11:IRSA,Paolini11:CSA_ICC,Kissling11:CRA,Popovski2012:Aloha,Pfister2012:IRSA} take advantage of time diversity to resolve collisions.

 In this paper, we develop and analyze a simple yet powerful \emph{relay-aided} \ac{SA} scheme which enjoys space diversity.
 More specifically,  $K$ independent observations of a slot are supposed to be available. The different observations are associated to $K$ relays, and, for each of them, the transmitted packets are subject to independent fading coefficients. Collisions are regarded as destructive, and the system is complemented by having relays deliver what they have decoded to a centralized gateway.

\ac{SA} with space (antenna) diversity was analyzed in \cite{Zorzi98:ALOHA} under the assumption of Rayleigh fading and shadowing, with emphasis on the two-antenna case. With respect to \cite{Zorzi98:ALOHA}, we introduce in our analysis a simplified channel model. In particular, the uplink wireless link connecting user $i$ and relay $j$ is described by a packet erasure channel with erasure probability $\varepsilon_{i,j}$, following the on-off fading model \cite{OnOff2003}. The fading is assumed to be independent for each pair of user-relay pair. Despite its simplicity, the model is accurate enough for some cases of interest. As an example, it captures the main features of an interactive satellite network with satellite located on different orbits, and where the line-of-sight link between users and relays may be blocked whenever an obstacle lies between a user and a satellite (here, the satellites play the role of relays).

Under this fading model, elegant exact expressions for the system throughput as a function of the number of relays are derived, yielding  deep insights in the gains provided by diversity in \ac{SA} protocols. We further provide an analysis on how the link between the relays and the centralized gateways (also referred to as \emph{downlink}) shall be dimensioned, assuming the relays to be uncoordinated. A bound on the downlink capacity is derived, which is achieved by a random linear coding approach based on Slepian-Wolf coding.

The rest of the paper is organized as follows. We start in Section~\ref{sec:systemModel} by defining the system model that is used to develop our framework. Section~\ref{sec:uplink} provides a thorough analysis of the system uplink, characterizing it in terms of throughput and delivery reliability, whereas in Section~\ref{sec:finiteDLCapacity} we study how to effectively deliver collected packets to a common gateway without resorting to coordination and information exchange among relays. In Appendix, we also investigate, for the two-receiver case, an extension of the considered scheme that takes advantage of successive interference cancellation techniques.

\section{System Model and Preliminaries} \label{sec:systemModel}

Throughout this paper, we focus on the topology depicted in Fig.~\ref{fig:simple_topology}, where an infinite population of users
want to deliver information in the form of data packets to a collecting gateway (GW). The transmission process is divided in two phases, referred to as \emph{uplink} and \emph{downlink}, respectively. During the former, data are sent in an uncoordinated fashion over a shared wireless channel to a set of $K$ receivers or relays, which, in turn, forward collected information to the GW in the downlink.

As to the uplink, time is divided in successive slots, and transmission parameters in terms of packet length, coding and modulation are fixed such that one packet can be sent within one time unit. Users are assumed to be slot-synchronized, and Slotted Aloha (SA) \cite{Abramson:ALOHA} is employed as medium access policy. Furthermore, the number of users accessing the channel in a generic slot is modelled as a Poisson-distributed r.v. $U$ of intensity $\rho$, with:
\begin{equation}
\textrm{Pr}\{ U = u \} = \frac{\rho^u e^{-\rho}}{u!}\,.
\end{equation}

\begin{figure}
\centering
\includegraphics[width=0.8\columnwidth]{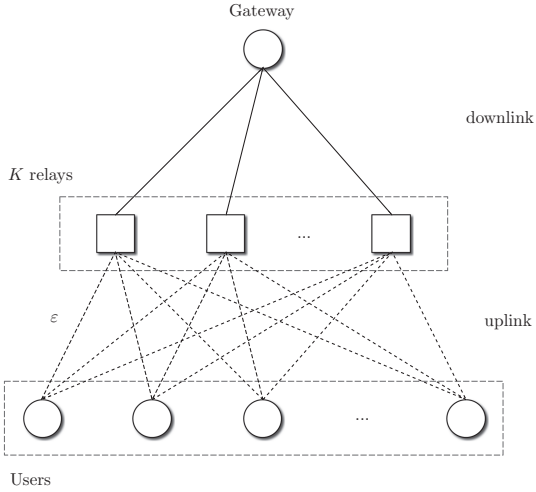}
\caption{Reference topology for the system under consideration.}
\label{fig:simple_topology}
\end{figure}

The uplink wireless link connecting user $i$ and receiver $j$ is described by a packet erasure channel with erasure probability $\varepsilon_{i,j}$, where independent realizations for any $(i,j)$ pair as well as for a specific user-receiver couple across time slots are assumed. For the sake of mathematical tractability, we set $\varepsilon_{i,j} = \varepsilon, \: \forall \,\, i,\, j$. Following the on-off fading description \cite{OnOff2003}, we assume that a packet is either completely shadowed, not bringing any power or interference contribution at a receiver, or it arrives unfaded. While, on the one hand, such a model is especially useful to develop mathematically tractable approaches to the aim of highlighting the key tradeoffs of the considered scenario, it also effectively captures effects like fading and short-term receiver unavailability due, for instance, to the presence of obstacles.
Throughout our investigation, no multi-user detection capabilities are considered at the relays, so that collisions among non-erased data units are regarded as destructive and prevent decoding at a receiver.

Within this framework, the number of non-erased packets that arrive at a relay when $u$ concurrent transmissions take place follows a binomial distribution of parameters $(u,1-\varepsilon)$ over one slot. Therefore, a successful reception occurs with probability $u(1-\varepsilon)\varepsilon^{u-1}$, and the average throughput experienced at each of the $K$ receivers, in terms of decoded packets per slot, can be computed as:
\begin{align}
\mathcal T_{sa} &= \sum\limits_{u=0}^{\infty} \frac{\rho^u e^{-\rho}}{u!} \, u(1-\varepsilon)\varepsilon^{u-1} = \rho (1-\varepsilon) e^{-\rho(1-\varepsilon)}\,,
\end{align}
corresponding to the performance of a SA system with erasures.
On the other hand, a spatial diversity gain can be triggered when the  relays are considered jointly, since independent channel realizations may lead them to retrieve different information units over the same time slot. In order to quantify this beneficial effect, we label a packet as \emph{collected} when it has been received by at least one of the relays, and we introduce the \emph{uplink throughput} $\mathcal T_{up,K}$ as the average number of collected packets per slot. Despite its simplicity, such a definition offers an effective characterization of the beneficial effects of diversity, by properly accounting for both the possibility of retrieving up to $\min\{u,K\}$ distinct data units or multiple times the same data unit over a slot, as will be discussed in details in Section~\ref{sec:uplinkThroughput}. On the other hand, $\mathcal T_{up,K}$ also quantifies the actual amount of information that can be retrieved by the set of receivers, providing an upper bound for the overall achievable end-to-end performance, and setting the target for the design of any relay-to-GW delivery strategy.

For the downlink phase, we focus on a \emph{decode and forward} (D\&F) approach, so that each receiver re-encodes and transmits only packets it has correctly retrieved during the uplink phase, or possibly linear combinations thereof. A finite downlink capacity is assumed, and relays have to share a common bandwidth to communicate to the GW by means of a TDMA scheme. In order to get an insightful characterization of the optimum achievable system performance, we assume relay-to-GW links to be error free, and let resource allocation for the D\&F phase be performed ideally and without additional cost by the central collecting unit.

We then complement our study in Appendix~\ref{app:SIC} by considering, for the simplified $K=2$ scenario, an \emph{amplify and forward} (A\&F) approach. In this case, relays simply deliver an amplified version of the analog waveform (possibly the outcome of a collision) they received, whereas the GW performs decoding relying on successive interference cancellation (SIC) techniques. The goal of such an investigation is to derive a characterization of the gains that are achievable by jointly processing signals incoming at different receivers. Along this line of reasoning, we will focus on an idealized downlink, such that information can be reliably delivered to the collecting unit at no cost in terms of bandwidth.

\subsection{Notation}
Prior to delving into the details of our mathematical framework, we introduce in the following some useful notation. All the variables will be properly introduced when needed in the discussion, and the present section is simply meant to offer a quick reference point throughout the reading.

$K$ relays are available, and, within time slot $t$, the countably infinite set of possible outcomes at each of them is labeled as $\Omega_t:=\{\omega_0^t,\omega_1^t,\omega_2^t,\ldots,\omega_\infty^t\}$ for each $t=1,2,\ldots,n$. Here, $\omega_0^t$ denotes the erasure event (given either by a collision or by an idle slot), while $\omega_j^t$ indicates the event that the packet of the $j$-th user arriving in slot $t$ was received. According to this notation, we define as
$X_k^t$ the random variables with alphabet $\mathbb N$, where $X_k^t = j$ if $\omega_j^t$ was the observation at relay $k$.
When needed for mathematical discussion, we let the uplink operate for $n$ time slots. In this case, let $\mc A_k^n$ be the set of collected packets after $n$ time slots at receiver $k$, where $\mc A_k^n \subsetneq \bigcup_{t=1}^n \{ \Omega_t \backslash \omega_0^t\}$. That is, we do not add the erasure events to $\mc A_k^n$. The number of received packets at relay $k$ after $n$ time slots is thus $|\mc A_k^n|$.

In general, the complement of a set $\mc A$ is indicated as $\overline{\mc A}$. We write vectors as lowercase underlined variables, e.g.,  $\ve w$, while matrices and their transposes are labeled by uppercase letters, e.g.,  $G$ and $G^T$.

\section{A Characterization of System Uplink} \label{sec:uplink}

With reference to the topology of Fig.~\ref{fig:simple_topology}, we first consider the uplink phase. In order to gather a comprehensive description of the improvements enabled by receiver diversity, we characterize the system by means of two somewhat complementary metrics: uplink throughput (Section~\ref{sec:uplinkThroughput}) and packet loss rate (Section~\ref{sec:PLR}).

\subsection{Uplink Throughput} \label{sec:uplinkThroughput}
Let us focus on the random access channel, and, following the definition introduced in Section~\ref{sec:systemModel}, let $C$ be the number of packets collected by the relays over one slot. $C$ is a r.v. with outcomes in the set $\{0,1,2, \ldots, K\}$, where the maximum value occurs when the $K$ receivers decode distinct packets due to different erasure patterns. The average uplink throughput can thus be expressed by conditioning on the number of concurrent transmissions as:
\begin{equation}
\mathcal T_{up,K} =\!\mathbb E_U [ \, \mathbb E[ \,C \,|\, U \,]\, ] \!=\!\! \sum\limits_{u=0}^\infty \frac{\rho^u e^{-\rho}}{u!}  \sum\limits_{c=0}^K c\,\textrm{Pr}\{C=c \,|\,U=u\}.
\label{eq:truUplink_general}
\end{equation}
While Eq.~(\ref{eq:truUplink_general}) formula holds for any $K$, the computation of the collection probabilities intrinsically depends on the number of available relays. In this perspective, we articulate our analysis by first considering the two-receiver case, to then extend the results for an arbitrary topology.

\subsubsection{The Two-Receiver Case}

Let us first then focus on the case in which only two relays are available. Such a scenario allows a compact mathematical derivation of the uplink throughput, as the events leading to packet collection at the relays set can easily be expressed. On the other hand, it also represents a case of practical relevance, as it can be instantiated by simply adding a receiver to an existing SA-based system.
\begin{figure}
\centering
\includegraphics[width=\columnwidth]{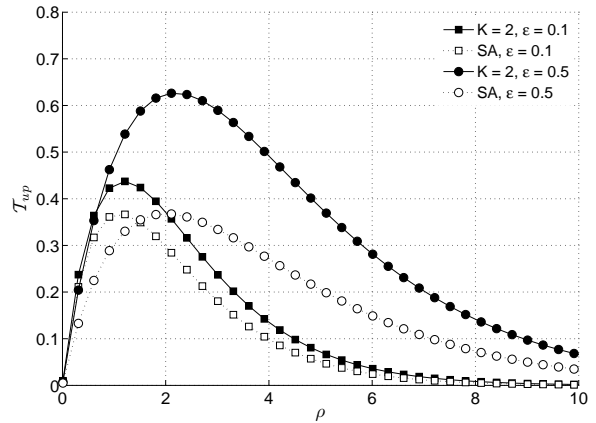}
\caption{Average uplink throughput vs channel load under different erasure probabilities. Black markers indicate the performance in the presence of two receivers, whereas white markers report the behavior of pure SA.}
\label{fig:truputUplink}
\end{figure}
When $K=2$, the situation for $C=1$ can easily be accounted for, since a single packet can be collected as soon as at least one of the relays does not undergo an erasure, i.e., with overall probability $1-\varepsilon^2$. On the other hand, by virtue of the binomial distribution of $U$, the event of collecting a single information unit over one slot occurs with probability
\begin{align}
\textrm{Pr}\{C =1 \,|\,U=u\} =& \,2u(1-\varepsilon) \varepsilon^{u-1}\left[ 1 - u(1-\varepsilon) \varepsilon^{u-1} \right]  \nonumber \\  &+ u(1-\varepsilon)^2 \varepsilon^{2(u-1)},
\end{align}
where the former addend accounts for the case in which one relay decodes a packet while the other does not (either due to erasures or to a collision), whereas the latter tracks the case of having the two relays decoding the same information unit. Conversely, a reward of two packets is obtained only when the receivers successfully retrieve distinct units, with probability
 \begin{equation}
\textrm{Pr}\{C=2 \,|\,U=u\} = u(u-1)(1-\varepsilon)^2 \varepsilon^{2(u-1)}.
\end{equation}
Plugging these results into (\ref{eq:truUplink_general}) we get, after some calculations, a closed-form expression for the throughput in the uplink and thus, as discussed, also for the end-to-end D\&F case with infinite downlink capacity:
\begin{equation}
\mathcal T_{up,2} = 2\rho (1-\varepsilon)\, e^{-\rho (1-\varepsilon)} - \rho (1-\varepsilon)^2 \, e^{-\rho (1-\varepsilon^2)}.
\label{eq:truUplink_closedForm}
\end{equation}
The trend of $\mathcal T_{up,2}$ is reported in Fig.~\ref{fig:truputUplink} against the channel load $\rho$ for different values of the erasure probability, and compared to the performance in the presence of a single receiver, i.e., $\mathcal T_{sa}$.
Eq.~(\ref{eq:truUplink_closedForm}) conveniently expresses $\mathcal T_{up,2}$ as twice the throughput of SA in the presence of erasures, reduced by a loss factor which accounts for the possibility of having both relays decode the same information unit. In this perspective, it is interesting to evaluate the maximum throughput $\mathcal T_{up,2}^{*}(\varepsilon)$ as well as the optimal working point $\rho^{*}(\varepsilon)$ achieving it for the system uplink. The transcendental nature of (\ref{eq:truUplink_closedForm}) does not allow to obtain a closed formulation of these quantities, which, on the other hand, can easily be estimated by means of numerical optimization techniques.
\begin{figure}
\centering
\includegraphics[width=\columnwidth]{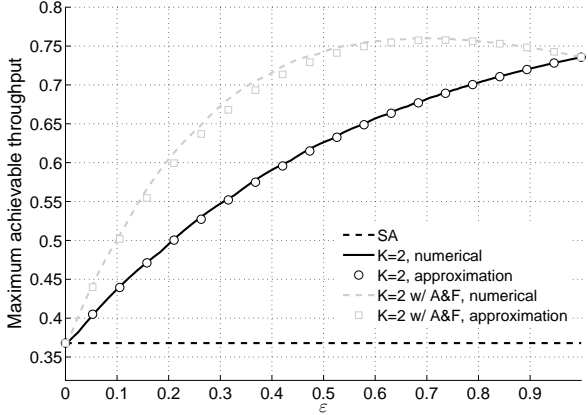}
\caption{Maximum uplink throughput vs erasure rate. The black continuous line reports the performance $\mc T^*_{up,2}$ of a two-receiver scheme, while white circled markers indicate $\mc T_{up,2}(1/(1-\varepsilon))$, and the dotted line shows the behavior of pure SA. Gray curves and markers are to be referred to the \emph{amplify and forward} case, that will be treated in Appendix~\ref{app:SIC}.}
\label{fig:maxTruVsEpsilon}
\end{figure}
The results of this analysis are reported in Fig.~\ref{fig:maxTruVsEpsilon}, where the peak throughput $\mathcal T_{up}^*$ is depicted by the black curve as a function of $\varepsilon$ and compared to the performance of SA, which clearly collects on average at most 0.36 pkt/slot regardless of the erasure rate. In ideal channel conditions, i.e., $\varepsilon=0$, no benefits can be obtained by resorting to multiple relays, as all of them would see the same reception set across slots. Conversely, higher values of $\varepsilon$ favour a decorrelation of the pattern of packets that can be correctly retrieved, and consequently improve the achievable throughput at the expense of higher loss rates. The result is a monotonically increasing behavior for $\mathcal T^*_{up,2}(\varepsilon)$, prior to plummeting with a singularity to a null throughput for the degenerate case $\varepsilon=1$. Fig.~\ref{fig:maxTruVsEpsilon} also reports (circled-white markers) the average throughput obtained for $\rho=1/(1-\varepsilon)$, i.e., when the
uplink
of the system under consideration operates at the optimal working point for a single-receiver SA, showing a tight match. In fact, even though the abscissa of maximum $\rho^*(\varepsilon)$ may differ from this value (they coincide only for the ideal case $\varepsilon=0$), the error which is committed when approximating $\mathcal T_{up,2}^*$ with $\mathcal T_{up,2}(1/(1-\varepsilon))$ can easily be shown numerically to never exceed 0.6\%, due to the very small slope of the function in the neighborhood of $\rho^*(\varepsilon)$. We can thus provide a very precise estimate of the peak uplink performance for a specific erasure rate as:
\begin{equation}
\mathcal T_{up,2}^*(\varepsilon) \simeq \frac{2}{e} - (1-\varepsilon)\, e^{-1-\varepsilon}, \hspace{.5cm} 0\leq \varepsilon < 1.
\label{eq:truUplink_peakApprox}
\end{equation}
which once again compactly captures the behavior of a two-receiver scenario by quantifying the loss with respect to twice the performance of SA.
In this perspective, two remarks shall be made. First of all, in order to approach the upper bound, the system has to be operated at very high load, as $\rho^* \simeq 1/(1-\varepsilon)$). These working points are typically not of interest, since very low levels of reliability can be provided by a congested channel with high erasure rates. Nevertheless, the presence of a second receiver triggers remarkable improvements already for loss probabilities that are of practical relevance, e.g., under harsh fading conditions or for satellite networks. Indeed, with $\varepsilon=0.1$ a $\sim 15\%$ raise can be spotted, whereas a loss rate of 20\% already leads to a 50\% throughput gain.
Secondly, the proposed framework highlights how no modifications in terms of system load are needed with respect to plain SA for a two-receiver system to be very efficiently operated. Such a result is particularly interesting, as it suggests that a relay node can be seamlessly and efficiently added to an already operating SA uplink when available, triggering the maximum achievable benefit without the need to undergo a re-tuning of the system which might be particularly expensive in terms of resources.

\subsubsection{The General Case, $K>2$}
\label{sec:general_throughput}
Let us now focus on the general topology reported in Fig.~\ref{fig:simple_topology}, where $K$ relays are available. While conceptually applicable, the approach presented to compute the uplink throughput in the two-receiver case becomes cumbersome as $K$ grows, due to the rapidly increasing number of events that have to be accounted for.
In order to characterize $T_{up,K}$, then, we follow a different strategy.
With reference to a single slot $t$, let $\Omega_t:=\{\omega_0^t,\omega_1^t,\omega_2^t,\ldots,\omega_\infty^t\}$ for each $t=1,2,\ldots,n$ be the countably infinite set of possible outcomes at each relay, where $\omega_0^t$ denotes the erasure event while $\omega_j^t$ indicates the event that the packet of the $j$-th user arriving in slot $t$ was received. Let us furthermore define as
$X_k^t$ the random variables with alphabet $\mathcal X=\{0,1,2,\ldots,\infty\}$, where $X_k^t = j$ if $\omega_j^t$ was the observation at relay $k$, so that $X_k^1,X_k^2,\ldots, X_k^n$ is an i.i.d. sequence for each relay $k$. We let the uplink operate for $n$ time slots, and indicate as $\mc A_k^n$ the set of packets collected at receiver $k$ over this time-span, where $\mc A_k^n \subsetneq \bigcup_{t=1}^n \{ \Omega_t \backslash \omega_0^t\}$ (i.e., we do not add the erasure events to $\mc A_k^n$). The number of received packets at relay $k$ after $n$ time slots is thus $|\mc A_k^n|$ and, with reference to this notation, we prove the following result:
\begin{prop}
 For an arbitrary number of $K$ relays, the throughput $\mathcal T_{up,K}$ is given by
 \begin{align}
 \mathcal T_{up,K} = \sum_{k=1}^K (-1)^{k-1} {K \choose k} \rho (1-\varepsilon)^k e^{-\rho(1-\varepsilon^k)}
\end{align}
\end{prop}
\vspace{2mm}
\begin{proof}
We have $|\mc A_k^n| = \sum_{t=1}^n \mathbbm{1}_{\{X_k^t\not = 0\}}$, where $\mathbbm{1}_{\{E\}}$ denotes the indicator random variable that takes on the value $1$ if the event $E$ is true and $0$ otherwise.
The throughput seen by a single relay can then be written as $\mathcal T_{up,1}=\mathbb E[\mathbbm{1}_{\{X_k^t\not = 0\}}] = \Pr\{X_k^t\not = 0\}$, and does not depend on the specific receiver being considered. By the weak law of large numbers,
\begin{align}
 \mathcal T_{up,1} = \lim_{n\rightarrow \infty} \frac{|\mc A_k^n| }{n}
\end{align}
or, more formally,
\begin{align}
 \lim_{n\rightarrow\infty} \Pr\left\{ \left| \frac{|\mc A_k^n| }{n} - \mathcal T_{up,1} \right| > \epsilon \right\} = 0 ~\text{for some }\epsilon > 0.
\end{align}
Similarly, for $K$ relays we have
\begin{align}
\mathcal T_{up,K} = \lim_{n\rightarrow\infty} \frac{|\bigcup_{k=1}^K \mc A_k^n|}{n}
\end{align}
By the inclusion-exclusion principle (see, e.g., \cite{slomson1991introduction}), we have
\begin{align}
 \left|\bigcup_{k=1}^K \mc A_k^n\right| = \sum_{\mc S\subseteq \{1,\ldots,K\}, \mc S \not = \emptyset} (-1)^{|\mc S|-1} \left| \mc I_{\mc S}^n \right|\\
 \text{with } \mc I_{\mc S}^n = \bigcap_{k \in \mc S} \mc A_k^n \label{eq:defI_S}
\end{align}
Here, $\mc I_{\mc S}^n$ denotes the set of packets that all the relay nodes specified by $\mc S = \{k_1,k_2,\ldots,k_{|\mc S|}\}$ have in common:
\begin{align}
 \left| \mc I_{\mc S}^n\right|=\left|\bigcap_{k \in \mc S} \mc A_k^n\right| = \sum_{t=1}^n \mathbbm{1}_{\{0 \not = X_{k_1}^t = X_{k_2}^t = \ldots = X_{k_{|\mc S|}}^t \} }
\end{align}
Due to symmetry in the setup, the value of $\left| \mc I_{\mc S}^n\right|$ only depends on the cardinality of $\mc S$ but not the explicit choice, so that $\left|\mc I_{\mc S}^n\right| = a_k^n$ for $k = |\mc S|$, and,
$$\left|\bigcup_{k=1}^K \mc A_k^n\right| = \sum_{k=1}^K (-1)^{k-1} {K \choose k} a_k^n.$$
As $X_k^1,X_k^2,\ldots, X_k^n$ are i.i.d., by the weak law of large numbers we have:
\begin{align}
 \lim_{n\rightarrow\infty}\frac{\left|\mc I_{\mc S}^n \right|}{n} =\Pr[\{0 \not = X_{k_1}^t = X_{k_2}^t = \ldots = X_{k_{|\mc S|}}^t \}].
\end{align}

We can compute the latter probability as
\begin{align}
\Pr\{0 &\not = X_{k_1}^t = \ldots = X_{k_{|\mc S|}}^t \} \nonumber\\
 = &\sum_{u} \Pr\{0 \not = X_{k_1}^t = \ldots = X_{k_{|\mc S|}}^t |U = u\}\Pr\{U=u\} \nonumber\\
 = & \sum_{u=1}^{\infty} \frac{e^{-\rho}\rho^u}{u!} {u \choose 1} \left((1-\varepsilon) \varepsilon^{u-1}\right)^{|\mc S|} \nonumber \\
 = & (1-\varepsilon)^{|\mc S|}\rho e^{-\rho(1-\varepsilon^{|\mc S|})}
\end{align}
As $\lim_{n\rightarrow\infty} \frac{a_k^n}{n} = (1-\varepsilon)^{k}\rho e^{-\rho(1-\varepsilon^{k})}$, the proposition follows.
\end{proof}

\begin{figure}
\centering
\includegraphics[width=\columnwidth]{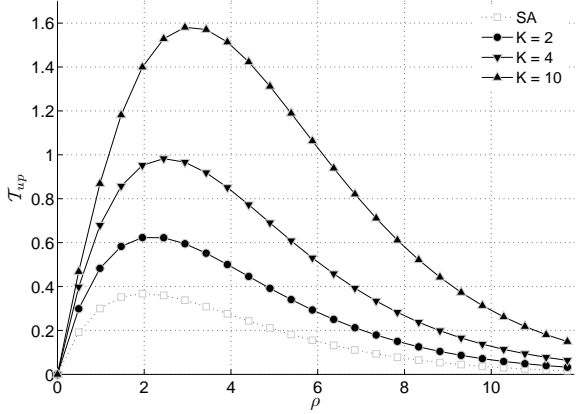}
\caption{Average uplink throughput vs channel load for different number of relays $K$. The erasure probability has been set to $\varepsilon=0.5$.}
\label{fig:truputUplink5rec}
\end{figure}

The performance achievable by increasing the number of relays is reported against the channel load in Fig.~\ref{fig:truputUplink5rec} for a reference erasure rate $\varepsilon=0.2$. As expected, $\mathcal T_{up,K}$ benefits from a higher degree of spatial diversity, showing how the system can collect more than one packet per uplink slot as soon as more than four receivers are available, for the parameters under consideration. Such a result stems from two main factors. On the one hand, increasing $K$ enables larger peak throughput over a single slot, as up to $K$ different data units can be simultaneously retrieved. On the other hand, broader receivers sets improve the probability of decoding packets in the presence of collisions even when less than $K$ users accessed the channel, by virtue of the different erasure patterns they experience.
The uplink throughput characterization is complemented by Fig.~\ref{fig:asymptoticULTru}, which reports the peak value for $\mathcal T^*_{up,K}$ (solid black curve), obtained by properly setting the channel load to $\rho^*_{K}$ (whose values are shown by the gray dashed curve), for an increasing relay population.\footnote{ As discussed for the $K=2$ case, a mathematical derivation of the optimal working point load $\rho^*_{K}$ is not straightforward, and simple numerical maximization techniques were employed to obtain the results of Fig~\ref{fig:asymptoticULTru}.} The plot clearly highlights how the benefit brought by introducing an additional receiver to the scheme, quantified by Eq.~(\ref{eq:incrementalGain}), progressively reduces, leading to a growth rate for the achievable throughput that is less than linear and that exhibits a logarithmic-like trend in $K$.
\begin{align}
\nonumber\Delta_{\mathcal T_{up}} &= \mathcal{T}_{up,K} - \mathcal{T}_{up,K-1} \\
&= \sum_{k=1}^K (-1)^{k-1} {K-1 \choose k-1} \rho (1-\varepsilon)^k e^{\rho (1-\varepsilon^k)}
\label{eq:incrementalGain}
\end{align}
\begin{figure}
\centering
\includegraphics[width=\columnwidth]{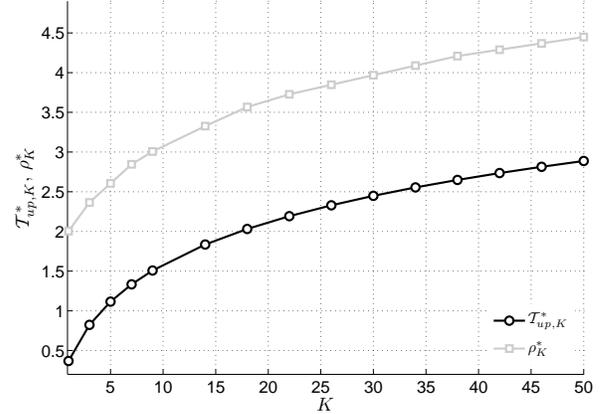}
\caption{Maximum achievable throughput $\mathcal T_{up,K}^*$ as a function of the number of relays $K$ for an erasure rate $\varepsilon=0.5$. The gray curve reports the load on the channel $\rho_{K}^*$ needed to reach $\mathcal T_{up,K}^*$.}
\label{fig:asymptoticULTru}
\end{figure}

\subsection{Packet loss probability} \label{sec:PLR}
The aggregate throughput derived in Section~\ref{sec:uplinkThroughput} represents a metric of interest towards understanding the potential of SA with diversity when aiming at reaping the most out of uplink bandwidth. On the other hand, operating an Aloha-based system at the optimal load $\rho^*_{K}$ exposes each transmitted packet to a loss probability that may not be negligible. In the classical single-receiver case without fading, for instance, the probability for a data unit not to be collected evaluates to $1-e^{-1}\simeq 0.63$. From this standpoint, in fact, several applications may resort to a lightly loaded random access uplink, aiming at a higher level of delivery reliability rather than at a high throughput. This is the case, for example, of channels used for logon and control signalling in many practical wireless networks. In order to investigate how diversity can improve performance in this direction, we extend our framework by computing the probability $\zeta_K$ that a user accessing the channel experiences a data loss, i.e., that the information unit it sends is not collected, either due to fading or to collisions, by any of the $K$ relays.

To this aim, let $\mc O$ describe the event that the packet of the observed user sent over time slot $t$ is not received by any of the receivers. Conditioning on the number of interferers $i$, i.e., of data units that were concurrently present on the uplink channel at $t$, the sought probability can be written as:
\begin{align}
 \zeta_K = \sum_{i=0}^\infty \Pr[\mc O|I=i] \Pr[I=i].
\end{align}
Here, the conditional probability can easily be determined recalling that each of the $K$ relays experiences an independent erasure pattern, obtaining $\Pr[\mc O|I=i] = (1-(1-\varepsilon)\varepsilon^{i})^K$ for an individual packet and $K$ relays with independent erasures on all individual links.
By resorting to the binomial theorem, such an expression can be conveniently reformulated as:
\begin{align}
 \Pr[\mc O|I=i] = \sum_{k=0}^K (-1)^k {K \choose k} \left((1-\varepsilon)\varepsilon^{i}\right)^k.
\end{align}
On the other hand, the number of interferers seen by a user that accesses the channel at time $t$ still follows a Poisson distribution of intensity $\rho$, so that, after simple calculations we finally get:
\begin{align}
 \zeta_K = \sum_{k=0}^K (-1)^k {K \choose k} (1-\varepsilon)^k e^{-\rho(1-\varepsilon^k)}.
 \label{eq:zeta_K}
\end{align}
\begin{figure}
\centering
\includegraphics[width=\columnwidth]{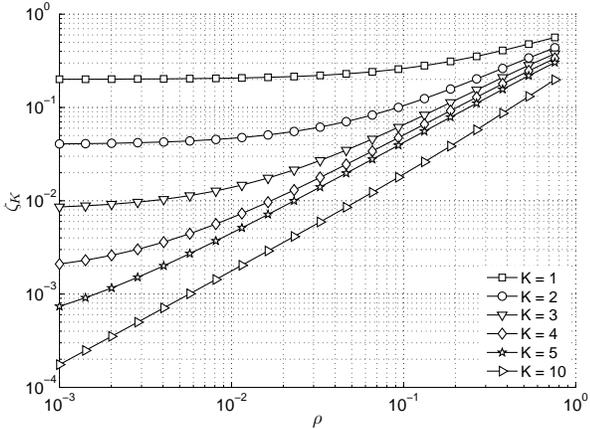}
\caption{Probability $\zeta_K$ that a packet sent by a user is not received by any of the relays. Different curves indicate different values of $K$, while the erasure probability has been set to $\varepsilon=0.2$.}
\label{fig:erasure_eff}
\end{figure}

Fig.~\ref{fig:erasure_eff} reports the behavior of $\zeta_K$ as a function of $\rho$ when the erasure rate over a single link is set to $\varepsilon = 0.2$. Different lines indicate the trend when increasing the number of receivers from 1 to 10. As expected, when $\rho\rightarrow 0$, a user accessing the channel is not likely to experience any interference, so that failures can only be induced by erasures, leading to an overall loss probability of $\varepsilon^K$. In this perspective, the availability of multiple receivers triggers a dramatic improvement, enabling levels of reliability that would otherwise not be possible irrespective of the channel configuration. On the other hand, Eq.~(\ref{eq:zeta_K}) turns out to be useful for system design, as it allows to determine the load that can be supported on the uplink channel while guaranteeing a target loss rate. Also in this case diversity can significantly ameliorate the performance. As shown in Fig.~\ref{fig:erasure_eff}, for example, a target loss rate $\zeta=5\cdot 10^{-2}$ is achieved by a three- and four-receiver scheme under 6- and 10-fold larger loads compared to the $K=2$ case, respectively.

 \section{Downlink Strategies} \label{sec:finiteDLCapacity}

The analysis carried out in Section~\ref{sec:uplink} has characterized the average number of packets that can be decoded at the relay set when SA is used in the uplink. We now instead consider the complementary task of delivering what has been collected to a central GW. In doing so, we aim at employing the minimum number of resources in terms of transmissions that have to be performed by the relays, while not allowing any information exchange among them. In particular, we consider a finite-capacity downlink, where the $K$ receivers share a common bandwidth to communicate with the GW by means of a time division multiple access scheme, and we assume that each of them can reliably deliver exactly one packet, possibly composed of a linear combination of what has been collected, over one time unit. We once again focus on a horizon of $n$ slots to operate the uplink, after which the downlink phase starts.

We structure our analysis in two parts. First, in Section~\ref{sec:bounds}, we derive lower bounds for the rates (in terms of downlink slots allocated per uplink slot) that have to assigned to receivers in order to deliver the whole set of data units collected in the uplink over the $n$ slots. Then, Section~\ref{sec:nc} shows how a simple forwarding strategy based on random linear network coding suffices to achieve optimality, completing the downlink phase in $\mathcal T_{up,K}$ slots for asymptotically large values of $n$.

Prior to delving into the details, let us introduce some useful notation. We denote the $L$-bit data part of packet of the $j$-th user arriving in time slot $t$ as $W_k^t \in \mc W$, with $\mc W = \mathbb{F}_{2^L} \cup e$, where $e$ is added as the erasure symbol. We furthermore assume that the receiver can determine the corresponding user through a packet header, i.e., the receiver knows both $j$ and $t$ after successful reception.
As the uplink operates over $n$ time slots, relay $k$ observes the vector $\ve w_k = [W_k^1, \ldots, W_k^n]$. In each time slot, the tuple $(W_1^t, W_2^t, \ldots, W_K^t)$ is drawn from a joint probability distribution $P_{W_1 \ldots W_K}$ which is governed by the uplink, and different relays might receive the same packet.

\subsection{Bounds for Downlink Rates} \label{sec:bounds}
Each relay $k$ transmits a packet in each of its $nR_k$ downlink slots.
We are interested in the set of rates $(R_k)_{k=1}^K$ such that the gateway can recover all packets (with high probability).

This is essentially the problem of distributed source coding (SW-Coding \cite{slepian1973noiseless}), with the following modification:
SW-coding ensures that the gateway can recover all $K$ observed strings $\ve w_k$, $k=1,2,\ldots,K$ perfectly.
In this setup, the gateway should be able to recover every packet that was received at any relay. However, neither is the gateway interested in erasures symbols at the relays, i.e. whenever $W_k^t=e$ for any $k$, $t$, nor in reconstructing each relay sequence perfectly.
The authors in \cite{dana2006capacity} overcame this problem by assuming that the decoder knows all the erasure positions of the whole network. This assumption applies in our case as packet numbers are supposed to be known via a packet header.
Let all erasure positions be represented by $\Gamma$.

The rates $(R_1, \ldots, R_K)$ are achievable \cite{slepian1973noiseless} if
\begin{align}
 \sum_{k \in \mc S} R_k \geq H(W_{\mc S}|W_{\overline{\mc S}},\Gamma), \quad \forall~ S\subseteq [1,2,\ldots, K]
\end{align}
where $ W_{\mc S} = (W_{k_1}^t, W_{k_2}^t, \ldots, W_{k_{|\mc S|}}^t)$,  denotes the observations at some time $t$ at the subset of receivers specified by $\mc S = \{k_1, k_2, \ldots, k_{|\mc S|}\}$.
$\Gamma$ has the effect of removing the influence of the erasure symbols on the conditional entropies.
Computing the entropies however requires the full probability distribution $P_{W_1 \ldots W_K}$ which is a difficult task in general. By different means, we can obtain the equivalent conditions:
\vspace{1mm}
\begin{prop}
 The rates $(R_k)_{k=1}^K$ have to satisfy
 \begin{align}
  \sum_{k \in \mc S} R_k \!\geq \mc T_{up,K} \!+\!\! \sum_{k=1}^{K-|\mc S|} (-1)^k {K-|\mc S| \choose k} \rho (1-\varepsilon)^k e^{-\rho (1-\varepsilon^k)},\nonumber\\ \forall \mc S \subseteq \{1,\ldots,K\}
 \end{align}
\end{prop}
\vspace{1mm}
\begin{proof}
Consider a subset of relays $\mc S \subseteq \{1,\ldots,K\}$ and their buffer contents $\bigcup_{k\in \mc S}\mc A_k^n$ after $n$ time slots.
 In order to satisfy successful recovery at the gateway, at least all packets that have been collected only by nodes in the set $\mc S$ and not by anyone else have to be communicated to the gateway. That is,
 \begin{align}
  \sum_{k \in \mc S} n\cdot R_k \geq \left| \bigcup_{k\in \mc S} \mc A_k^n \backslash \bigcup_{k \in \overline{\mc S}} \mc A_k^n \right|,
  \label{eq:ratebound1}
 \end{align}
 with $\overline{\mc S} = \{1,\ldots,K\} \backslash S$.
 Note that $$\bigcup_{k\in \mc S} \mc A_k^n \backslash \bigcup_{k \in \overline{\mc S}} \mc A_k^n = \mc A^n\backslash \bigcup_{k \in \overline{\mc S}} \mc A_k^n,$$ so by the inclusion-exclusion principle and due to $|\overline{\mc S}| = K - |\mc S|$
 \begin{align}
  \left|\bigcup_{k\in \mc S} \mc A_k^n \backslash \bigcup_{k \in \overline{\mc S}} \mc A_k^n \right| = |\mc A_K^n| + \sum_{k=1}^{K-|\mc S|} (-1)^k {K-|\mc S| \choose k} a_k^n
 \end{align}
with $\left|\mc I_{\mc S}^n\right| = a_k^n$ for $k = |\mc S|$ as before.
By plugging in the value for $\lim_{n\rightarrow\infty} \frac{a_k^n}{n}$, the proposition follows.
\end{proof}

\vspace{-1mm}
\subsection{Random Linear Coding} \label{sec:nc}
By means of Proposition 2, we have derived a characterization of the rates that have to be assigned to relays in order to deliver the whole set of collected packets to the GW. In this section, we complete the discussion by proposing a strategy that is capable of matching such conditions, thus achieving optimality.
The solution that we employ is based on a straight-forward application of the well-known random linear coding scheme in \cite{ho2006random}, and will therefore only briefly sketched in the following.

Each relay $k$ generates a matrix $G_k \in \mathbb F_{2^L}^{nR_k \times n}$ and obtains the data part of its $nR_k$ transmit packets by $\ve c_k^T = G_k \ve w_k^T$.
Whenever an element of $\ve w_k$ was an erasure symbol, the corresponding column of $G_k$ is an all-zero column. Erasure symbols thus have no contributions to the transmit packets $\ve c_k$. All other elements of $G_k$  are drawn uniformly at random from $\mathbb F_{2^L}^*$, where $\mathbb F_{2^L}^*$ denotes the multiplicative group of $\mathbb F_{2^L}$.

The gateway collects all incoming packets and obtains the system of linear equations
\begin{align}
 \underbrace{\left( \begin{array}{c} \ve c_1^T\\ \ve c_2^T \\ \vdots \\ \ve c_K^T \end{array} \right)}_{\ve c^T} =
 \underbrace{\left(\begin{array}{cccc} G_1 & 0 & \ldots & 0\\0 & G_2 & \ldots & 0 \\ 0 & 0 & \ddots & 0 \\ 0 & 0 & \ldots & G_K \end{array}\right)}_{G}
 \underbrace{\left( \begin{array}{c} \ve w_1^T\\ \ve w_2^T \\ \vdots \\ \ve w_K^T \end{array} \right)}_{\ve w^T}
\end{align}
where $G\in \mathbb F_{2^L}^{n\sum_k R_k \times nK}$. Note that some elements of $\ve w$ can be identical because they were received by more than one relay and thus are elements of some $\ve w_{k_1}$, $\ve w_{k_2}, \ldots$.
One can merge these entries in $\ve w$ that appear more than once.
Additionally, we drop all erasure-symbols in $\ve w$ and delete the corresponding columns in $G$ to obtain the reduced system of equations
\begin{align}
 \ve c^T = \tilde G \ve{\tilde w}^T
\end{align}
where $\ve{\tilde w} \in \mathbb F_{2^L}^{|\mc A^n|}$ contains only distinct received packets and no erasure symbols.
Clearly, there are $\left| \bigcup_{k\in \mc S} \mc A_k^n \right| $ elements in $\ve{\tilde w}$.

We partition the entries in $\ve{\tilde w}$ into $2^K-1$ vectors $\ve{\tilde w}_{\mc S}$ for each nonempty subset $S\subseteq \{1,2,\ldots,K\}$: Each vector $\ve{\tilde w}_{\mc S}$ contains all packets that have been received only by all relays specified by $\mc S$ and not by anyone else.
That is, $\ve{\tilde w}_{\mc S}$ corresponds to the set $\mc P_{\mc S}^n = \bigcap_{k\in \mc S} \mc A_k^n \backslash \bigcup_{k \in \overline{\mc S}} \mc A_k^n $, its length is $|\mc P_{\mc S}^n|$.

The columns in $\tilde G$ and rows in $\ve{\tilde w}^T$ can be permuted such that one can write
\begin{align}
 \ve c_k^T = \tilde G_k \ve{\tilde w}:= \sum_{\mc S\subseteq \{1,2,\ldots,K\}} \tilde{G}_{k,\mc S} \cdot \ve{\tilde w}_{\mc S}, \quad \forall~k=1,\ldots,K.
\end{align}
Each of the matrices $\tilde{G}_{k,\mc S} \in \mathbb F_{2^L}^{nR_k \times |\mc P_{\mc S}^n|}$ contains only elements from $\mathbb F_{2^L}^*$ if $k\in S$ and is an all-zero matrix otherwise. A compact representation for $K=3$ is shown in (\ref{eq:exampleG_K=3}) at the bottom of next page.

The variables that are involved only in $n \sum_{k \in \mc S}R_k$ equations are those in $\ve{\tilde w}_{\mc L},~ \mc L \subseteq \mc S$, for each subset $\mc S\subseteq \{1,2,\ldots, K\}$.
For decoding, the number of equations has to be larger or equal to the number of variables, so a necessary condition for decoding is that $n \sum_{k \in \mc S}R_k \geq \sum_{\mc L \subseteq \mc S}  \left|\mc P_{\mc L}^n\right|$. This is satisfied by (\ref{eq:ratebound1}), since $\sum_{\mc L \subseteq \mc S}  \left|\mc P_{\mc L}^n\right| = \left|\bigcup_{k\in \mc S} \mc A_k^n \backslash \bigcup_{k \in \overline{\mc S}} \mathcal A_k^n   \right|$, as we show in Appendix~\ref{app:connection}

A sufficient condition is that the matrix $\tilde G_k,~ k\in \mc S$
representing $n \sum_{k \in \mc S}R_k$ equations has rank $\sum_{\mc L \subseteq \mc S}  \left|\mc P_{\mc L}^n\right|$
for each subset $\mc S\subseteq \{1,2,\ldots, K\}$.
Denote the set of indices of nonzero columns of matrix $\tilde G_k$ as the support of $\tilde G_k$.
Note that a row of matrix $\tilde G_k$ has a different support than a row of matrix $\tilde G_l$, for $k\not = l$. These rows are thus linearly independent. It thus suffices to check that all rows of matrix $\tilde G_k$ are linearly independent. As all nonzero elements are randomly drawn from $\mathbb F_{2^L}^*$, the probability of linear dependence goes to zero as $L$ grows large, completing our proof, and showing that the presented forwarding scheme achieves the bounds of Proposition 2.

\begin{figure*}[!bp]
\normalsize
\hrulefill
\begin{align}
 \left( \begin{array}{c} \ve c_1^T\\ \ve c_2^T \\  \ve c_3^T \end{array} \right) =
  \left(\begin{array}{c}\tilde G_1\\ \tilde G_2 \\ \tilde G_3
                            \end{array}\right)
                            \ve{\tilde w}^T
 =
 \left(\begin{array}{ccccccc} \tilde G_{1,\{1\}} & 0 & 0 & \tilde G_{1,\{1,2\}} & \tilde G_{1,\{1,3\}} & 0 & \tilde G_{1, \{1,2,3\}}
                            \\0 & \tilde G_{2,\{2\}} & 0 & \tilde G_{2,\{1,2\}} & 0 & \tilde G_{2, \{2,3\}} & \tilde G_{2, \{1,2,3\}}
                            \\0 & 0 & \tilde G_{3,\{3\}} & 0 & \tilde G_{3, \{1,3\}} & \tilde G_{3, \{2,3\}} & \tilde G_{3, \{1,2,3\}}
                            \end{array}\right)
 \left( \begin{array}{l} \ve{\tilde w}_{\{1\}}^T\\ \ve{\tilde w}_{\{2\}}^T \\ \ve{\tilde w}_{\{3\}}^T \\ \ve{\tilde w}_{\{1,2\}}^T \\ \ve{\tilde w}_{\{1,3\}}^T \\ \ve{\tilde w}_{\{2,3\}}^T \\ \ve{\tilde w}_{\{1,2,3\}}^T
 \end{array} \right)
 \label{eq:exampleG_K=3}
\end{align}
\end{figure*}

%
%
%
%

%

\section{Conclusions}\label{sec:conclusions}

In this paper, a simple and practical extension of Slotted Aloha in the presence of multiple receivers, or relays, has been presented and thoroughly discussed. By means of an analytical framework, closed-form expression for the uplink throughput (defined as the average number of packets per slot \emph{collected} by the set of $K$ relays), as well as for the probability that a data unit is not retrieved by any of the receivers, have been derived for an arbitrary value of $K$ under the assumption of on-off fading. Remarkable gains have been shown and discussed already for a moderate number of receivers. The study is complemented by considering the problem of delivering the set of collected packets to a common gateway without allowing any information exchange among receivers. Theoretical bounds for the amount of resources that have to be allocated to achieve this task have been derived, and a simple scheme based on random linear network coding has been shown to match such bounds.

\appendices
\section{Amplify and Forward with SIC at the Gateway}
\label{app:SIC}

The framework developed in this paper has focused on characterizing the performance achievable by a SA system with receiver diversity when a D\&F scheme is implemented at intermediate nodes. On the other hand, restricting relays to simply forward what they have successfully retrieved in the uplink prevents the GW from performing joint decoding on possibly uncorrelated signals. In order to go beyond this limitation, we consider in this appendix the possibility for receivers to send in the downlink, on a slot-basis, an amplified version of the analog waveform they perceive even in the presence of a collision, following an \emph{amplify and forward} (A\&F) approach. For the sake of mathematical tractability, we focus on the $K=2$ case, and we model relays to instantly and reliably deliver information to the GW. The advantage of such an assumption is twofold. On the one hand, it will allow us once more to identify elegant closed-form expressions for the throughput of A\&F with spatial diversity, highlighting the fundamental tradeoffs that arise in the presence of multiple receivers. On the other hand, despite its ideality, the model under consideration is representative for several scenarios of practical interest, in which the bandwidth available in the downlink is much larger than the one of the uplink. Satellite networks, as well as topologies where multiple base stations or access points are connected to a coordinating unit via a wideband backbone may be examples in this direction.

At the gateway side, successive interference cancellation (SIC) is applied to the collected signals. This approach offers an improvement whenever the waveform forwarded by one of the relays allows decoding of a packet, say $x$, while the other reports a collision given by the superposition of $x$ and one other packet. In this condition, the set of relays would be able to collect only one information unit, whereas, with A\&F, the gateway can subtract the interference contribution of $x$ from the
collision-corrupted waveform and successfully collect the second packet as well. Details on the accuracy of this model on noisy channels with actual signal processing techniques can be found in \cite{DeGaudenzi07:CRDSA,Liva11:IRSA}.

The gain offered by SIC can thus be computed for the $K=2$ case by simply adding to the uplink throughput derived in Section~\ref{sec:uplinkThroughput} one additional collected data unit each time the described collision condition is met. Hence, we can write:
\begin{align}
\mathcal G_{SIC} =&\, \sum\limits_{u=2}^\infty \frac{\rho^u \,e^{-\rho}}{u!} \, 2u (u-1) \, \varepsilon^{u-1} (1-\varepsilon) \cdot \varepsilon^{u-2} (1-\varepsilon)^2\nonumber\\ =&\, 2 \rho^2 \varepsilon (1-\varepsilon)^3 \, e^{-\rho(1-\varepsilon^2)},
\end{align}
where, within the summation, $\varepsilon^{u-1} (1-\varepsilon)$ accounts for the correct reception at one relay while $\varepsilon^{u-2} (1-\varepsilon)^2$ enforces a collision of exactly two packets at the other relay, for a total of $2u(u-1)$ configurations that can be solved with SIC. The average number of collected packets at the GW per uplink slot, which we refer to as $\mathcal T_{A\&F}$, is thus simply expressed as $\mathcal T_{A\&F} = \mathcal T_{up,2}+ \mathcal G_{SIC}$:
\begin{align}
\mathcal T_{A\&F}
=& \, 2\rho(1-\varepsilon)\, e^{-\rho(1-\varepsilon)} - \rho(1-\varepsilon)^2\, e^{-\rho(1-\varepsilon^2)} \nonumber \\
+& 2 \rho^2 \varepsilon (1-\varepsilon)^3 \, e^{-\rho(1-\varepsilon^2)}.
\end{align}

\begin{figure}
\centering
\includegraphics[width=\columnwidth]{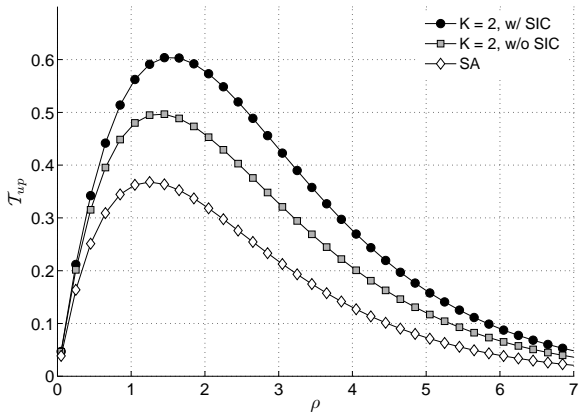}
\caption{End-to-end throughput vs uplink channel load with infinite downlink capacity and $\varepsilon=0.2$. The black-marked curve reports the behavior with A\&F; the gray-marked curve indicates D\&F; the white-marked curve shows the behavior of SA with a single receiver.}
\label{fig:e2eThroughputSIC}
\end{figure}

The obtained trend is plotted in Fig.~\ref{fig:e2eThroughputSIC} against the channel load $\rho$ for an erasure probability of $\varepsilon=0.2$, showing a 20\% and 66\% improvement in peak throughput compared to the performance achieved under the same conditions in the uplink (i.e., without SIC) and by a SA scheme with single receiver, respectively.
As discussed in Section~\ref{sec:uplinkThroughput}, a closed-form evaluation of the maximum throughput $\mathcal T_{A\&F}^*$ is not straightforward, due to the transcendental nature of the terms that define the metric. Nevertheless, in the two-receiver case, a good approximation is once again offered by evaluating $\mathcal T_{A\&F}$ at $\rho = 1/(1-\varepsilon)$, obtaining, after some calculations:
\begin{equation}
\mathcal T_{A\&F}^*(\varepsilon) \simeq  \frac{2}{e} - e^{-1-\varepsilon}\,(1-3\varepsilon+2\varepsilon^2),
\end{equation}
where a loss factor of $e^{-1-\varepsilon}(1-\varepsilon)(1-2\varepsilon) \leq e^{-1-\varepsilon}(1-\varepsilon)$ is exhibited with respect to the upper bound provided by twice the throughput of SA. The behavior of $\mathcal T_{A\&F}^*(\varepsilon)$ is reported in Fig.~\ref{fig:maxTruVsEpsilon}, where the dashed-gray curve indicates the actual peak throughput values computed numerically, whereas the white-squared markers report the proposed approximation. The plot highlights how, as opposed to what discussed for the non-SIC case, the introduction of joint processing modifies the shape of the curve, identifying an optimal (albeit not practical for many applications) erasure probability where the throughput is more than doubled over SA. On the other hand, it is remarkable to point out that the most relevant improvements over a non-SIC multi-receiver uplink are triggered exactly for values of $\varepsilon$ that may indeed be experienced in practical scenarios, boosting the peak throughput by up to 25\%.

As a concluding observation, notice how the two-relay A\&F solution that we discussed requires rather simple interference cancellation procedures compared to other advanced random access schemes \cite{DeGaudenzi07:CRDSA,Liva11:IRSA}, as only two observations need to be considered for joint decoding. The presented architecture, thus, represents an interesting tradeoff between complexity and performance gain, and triggers interest in more advanced scenarios where more receivers are available.

\section{}
\label{app:connection}
We derive that
$\sum_{\mc L \subseteq \mc S}  \left|\mc P_{\mc L}^n\right| = \left|\bigcup_{k\in \mc S} \mc A_k^n \backslash \bigcup_{k \in \overline{\mc S}} \mathcal A_k^n   \right|$.
We need the following lemma.
\begin{lemma}
 For a collection of sets $\mc B_1,\mc B_2,\ldots, \mc B_K$ and a subset $\mc S\subseteq \{1,\ldots,K\}$,
\begin{align}
 \bigcup_{k\in \mc S} \mc B_k = \bigcup_{\mc L \subseteq \mc S}\left( \bigcap_{l\in \mc L} \mc B_l \backslash \bigcup_{s \in \mc S\backslash \mc L} \mc B_s\right),
\end{align}
where the sets on the RHS do not intersect and thus form a partition of the LHS.
\end{lemma}
\begin{proof}
 We first show that any element $b\in  \bigcup_{k\in \mc S} \mc B_k$ is also included in the RHS:
 Pick an element $b\in  \bigcup_{k\in \mc S} \mc B_k$. Assume $\mc L$ is the subset of largest cardinality such that $b \in \mc B_l,~\forall l \in \mc L$. Clearly, $b \in \bigcap_{l\in\mc L}\mc B_l$ but $b \not \in \bigcup_{s\in \mc S\backslash \mc L} \mc B_s$. It follows that $b \in \bigcap_{l\in\mc L}\mc B_l \backslash \bigcup_{s\in \mc S\backslash \mc L} \mc B_s$. This is true for some subset $\mc L \subseteq \mc S$.
Second, we show that this subset is unique.
Let again $\mc L$ be the subset of largest cardinality such that $b \in \mc B_l,~\forall l \in \mc L$ and choose a different subset $\mc V\subseteq \mc S$, $\mc V \not = \mc L$. Then, either $b \not \in \bigcap_{l \in \mc L} B_l$ or $b\in \bigcup_{s \in \mc S \backslash \mc L} \mc B_s$. The element $b$ is thus only included in $\bigcap_{l\in\mc L}\mc B_l \backslash \bigcup_{s\in \mc S\backslash \mc L} \mc B_s$.
\end{proof}
By choosing $\mc B_k = \mc A_k^n \backslash \bigcup_{k \in \bar{\mc S}} \mc A_k^n$, the result follows by elementary set operations.

\section*{Acknowledgement}
The authors would like to thank Prof. Gerhard Kramer for the fruitful and insightful discussions throughout the development of the present work.
\bibliography{IEEEabrv,aloha}

\end{document}